\documentclass[review]{elsarticle}
\usepackage[latin1]{inputenc} 
\usepackage{listings, tcolorbox}

\usepackage{lineno,hyperref}
\usepackage{amssymb}
\usepackage{color}
\usepackage{amsmath}
\usepackage{tikz}
\usepackage{amsfonts}
\usepackage{mathrsfs}
\usepackage{amsthm}

\usepackage{bbding}
\usepackage{mathrsfs}
\usepackage{amsmath, amsfonts, amssymb, mathrsfs, txfonts}
\usepackage{graphicx, subfigure}
\usepackage{wasysym}
\usepackage{color}
\usepackage{bm}
\usepackage{enumerate}

\usepackage{pifont}
\usepackage{txfonts}
\usepackage{amsmath}
\usepackage{graphicx}
\usepackage{amsfonts}
\usepackage{amssymb}
\usepackage{mathrsfs,psfrag,eepic,epsfig}
\usepackage{epstopdf}

\biboptions{numbers,sort&compress}
\newtheorem{thm}{Theorem}[section]

\newtheorem{prop}[thm]{Proposition}
\theoremstyle{definition}

\modulolinenumbers[5]

\journal{Journal of \LaTeX\ Templates}

\makeatletter \@addtoreset{equation}{section}
\renewcommand{\theequation}{\arabic{section}.\arabic{equation}}

\begin{document}

\begin{frontmatter}

\title{Inverse scattering transform and multi-solition solutions for the sextic nonlinear Schr\"{o}dinger equation}
\tnotetext[mytitlenote]{Project supported by the Fundamental Research Fund for the Central Universities under the grant No. 2019ZDPY07.\\
\hspace*{3ex}$^{*}$Corresponding author.\\
\hspace*{3ex}\emph{E-mail addresses}: xwu@cumt.edu.cn (X. Wu), sftian@cumt.edu.cn,
shoufu2006@126.com (S. F. Tian),
and jinjieyang@cumt.edu.cn (J.J. Yang)}

\author{Xin Wu, Shou-Fu Tian$^{*}$ and Jin-Jie Yang}
\address{
School of Mathematics and Institute of Mathematical Physics, China University of Mining and Technology, Xuzhou 221116, People's Republic of China
}

\begin{abstract}
In this work, we consider the inverse scattering transform and  multi-solition solutions  of the sextic nonlinear Schr\"{o}dinger equation. The Jost functions of spectrum problem are derived directly, and the scattering data with $t=0$ are obtained according to analyze the symmetry and other related properties of the Jost functions. Then we take use of translation transformation to get the relation between potential and kernel, and recover potential according to Gel'fand-Levitan-Marchenko (GLM) integral equations. Furthermore, the time evolution of scattering data is considered, on the basic of that, the multi-solition solutions are derived. In addition, some solutions of the equation are analyzed and revealed its dynamic behavior via graphical analysis, which could be enriched the nonlinear phenomena of the sextic nonlinear Schr\"{o}dinger equation.
\end{abstract}


\end{frontmatter}


\section{Introduction}

It is well-known that the nonlinear Schr\"{o}dinger (NLS) equation is a nonlinear equation with soliton solutions, which plays an important role in mathematical physics. The NLS equation reads
\begin{align}
iq_{t}+\frac{1}{2}q_{tt}+|q|^{2}q=0
\end{align}
in dimensionless form. The NLS equation has been investigated more and more deeply since Zakharov and Shabat \cite{Zakharov-1972} have done important research on it. In addition to quantum mechanics, the NLS equation can also be used to describe all kinds of nonlinear waves in physics, such as the propagation of laser beam in the medium with refractive index and amplitude related, the free water wave of ideal fluid, plasma wave, etc.

However, the only basic model is not enough to study the new phenomena with the deepening of people's research on equations, especially when the wave amplitude increases. For example, in optics, one should introduce the high-order effects when considering pulses of shorter duration propagating along a fiber \cite{Chowdury-2015}-\cite{Chen-2016}.
Then a new form of NLS system is proposed which is the sextic nonlinear Schr\"{o}dinger equation \cite{Ankiewicz-2016}
\begin{gather} \label{Q1}
\begin{split}
&iq_{t}+\frac{1}{2}(q_{xx}+2|q|^{2}q)+\delta q_{xxxxxx}+\delta\left[60q^{*}|q_{x}|^{2}+50(q^{*})^{2}q_{xx}+2q^{*}_{xxxx}\right]q^{2}\\
&+\delta q\left[12q^{*}q_{xxxx}+8q_{x}q^{*}_{xxx}+22|q_{xx}|^{2}\right]
+\delta q\left[18q_{xxx}q^{*}_{x}+70(q^{*})^{2}q^{2}_{x}\right]+20\delta(q_{x})^{2}q^{*}_{xx}\\
&+10\delta q_{x}\left[5q_{xx}q^{*}_{x}+3q^{*}q_{xxx}\right]
+20\delta q^{*}q^{*}_{xx}+10\delta q^{3}\left[(q^{*}_{x})^{2}+2q^{*}q^{*}_{xx}\right]+20\delta q|q|^{6}=0,
\end{split}
\end{gather}
where $q(x,t)$ represents the complex envelops of the waves, $x$ and $t$ denote propagation distance and scaled time, the symbol `*' denotes the complex conjugation, and the real parameter $\delta$ denotes the coefficient of the sextic-order dispersion $q_{xxxxxx}$. When $\delta=0$, \eqref{Q1} will be reduced to the one-dimensional NLS equation.  Some work of the sixtic-order NLS equation have been investigated, such as breather solutions with imaginary eigenvalues and the interactions between two solitons of the sextic NLS equations \eqref{Q1} have been obtained in \cite{Sun-2017}.

It is generally known that the inverse scattering transform is a important and powerful analytical tool to solve integrable systems, which plays an indispensable role in the field of nonlinear sciences. The inverse scattering transform method was used to solve the Korteweg-de Vries (KdV) equation for the first time by Garderner, Greene, Kruskal and Miurra (GGKM) in 1967 \cite{Gardner-1967}, and optimized by Lax raised according to GGKM and form systematic inverse scatting transform method in 1968 \cite{Lax-1968}. The work of Zakharov and Shabat on the nonlinear Schr\"{o}dinger equation in 1972 \cite{Zakharov-1972} pushed the inverse scattering theory to a more general study. Since then, the inverse scattering method has been developed and extended to many excellent study and results \cite{1}-\cite{20}.
However, according to what we know, the inverse scatting transform of the sextic nonlinear Schr\"{o}dinger equation \eqref{Q1}   has not been studied. Therefore, the main purpose of this work is to find its inverse scatting transform, more abundant soliton solutions, and reveal the propagation behavior of the solutions.

The structure of this work is given as follows. In section 2, we derive the relevant Jost functions by analyzing the spectrum problem of the equation in detail. On the basis of the Jost functions, we get the scattering data when time is ignored temporarily. In section 3, GLM integral equations are established based on the integral kernel and the related definitions. The relationship among  some scattering data considering time evolution are also derived. Furthermore, we obtain the multi-soliton solutions according to the above results. In section 4, one-  and two- soliton solution are given, and their dynamic behaviors are analyzed in the form of figures. Finally, some conclusions and discussions are presented synoptically in the final section.

\section{Direct scattering transform}

In this section, the direct scattering transform of \eqref{Q1} will be studied and the scattering data can be obtained fanally. The Lax pair of the sextic nonlinear Schr\"{o}dinger equation reads
\addtocounter{equation}{1}
\begin{align}
&\psi_{x}=U\psi,\tag{\theequation a}\label{Q2}\\
&\psi_{t}=V\psi,\tag{\theequation b}
\end{align}
here $U$ and $V$ can be expressed in the form of
\begin{align*}
U=i\left(\begin{array}{cc}
    \lambda  & q^{*} \\
    q & -\lambda  \\
  \end{array}\right),
V=\sum_{c=0}^{6}i\lambda^{c}\left(\begin{array}{cc}
    A_{c} & B^{*}_{c} \\
    B_{c} & -A_{c}  \\
 \end{array}\right),
\end{align*}
with
\begin{align*}
A_{0}=&-\frac{1}{2}|q|^{2}-10\delta|q|^{6}-5\delta\left[q^{2}_{x}(q^{*})^{2}\right]
-10\delta|q|^{2}\left(q_{xx}q^{*}+q^{*}_{xx}q\right)-\delta|q_{xx}|^{2}\\
&+\delta\left(q_{x}q^{*}_{xxx}-q^{*}q_{xxxx}-qq^{*}_{xxxx}\right),\\
A_{1}=&12i\delta|q|^{2}\left(q_{x}q^{*}-q^{*}_{x}q\right)
+2i\delta\left(q_{x}q^{*}_{xx}-q^{*}_{x}q_{xx}+q^{*}q_{xxx}-q^{*}_{xxx}q\right),\\
A_{2}=&1+12\delta|q|^{4}-4\delta|q_{x}|^{2}+4\delta\left(q^{*}_{xx}q+q_{xx}q^{*}\right), A_{3}=8i\delta\left(qq^{*}_{x}-q^{*}q_{x}\right),\\
 A_{4}=&-16\delta|q|^{2}, A_{5}=0, A_{6}=32\delta,\\
B_{0}=&\frac{i}{2}q_{x}+i\delta q_{xxxxx}+10i\delta\left(qq^{*}_{x}q_{xx}+qq^{*}_{xx}q_{x}
+|q|^{2}q_{xxx}+3|q|^{4}q_{x}+q_{x}|q_{x}|^{2}2q^{*}q_{x}q_{xx}\right),\\
B_{1}=&q+12\delta q^{*}q^{2}_{x}+16\delta|q|^{2}q_{xx}+4\delta q^{2}q^{*}_{xx}+2\delta q_{xxxx}
+12\delta|q|^{4}q+8\delta q|q_{x}|^{2},\\
B_{2}=&-24i\delta|q|^{2}q_{x}-4i\delta q_{xxx}, B_{3}=-16\delta|q|^{2}q-8\delta q_{xx},\\
B_{4}=&16i\delta q_{x},
B_{5}=32\delta q,
B_{6}=0,
\end{align*}
where $\psi=\left(\psi_{1}(x,t),\psi_{2}(x,t)\right)^{T}$, and $\lambda$ is a  spectral parameter. Eq. \eqref{Q1} satisfies zero curvature equation $U_{t}-V_{x}+[U,V]=0$, which is the compatibility condition of the Lax pair (2.1).

According to the first expression of Lax pair, we can obtain the spectral problem as the following form
\begin{align}\label{Q3}
\left\{\begin{aligned}
&\psi_{1,x}(x,\lambda)=i\lambda\psi_{1}(x,\lambda)+q^{*}\psi_{2}(x,\lambda),\\
&\psi_{2,x}(x,\lambda)=-i\lambda\psi_{2}(x,\lambda)+q\psi_{1}(x,\lambda).
\end{aligned}\right.
\end{align}
Due to $q(x)$ belongs to Schwarz space, $q(x)$ and its derivatives with respect to $x$ approaches to zero rapidly when $|x|\rightarrow \infty$.

\subsection{Jost functions}

Because $t$ is a fixed variable, it is temporarily omitted.
\begin{thm}
The equations (2.2) have Jost functions which have the following integral forms
\addtocounter {equation}{1}
\begin{align}
&\phi_{1}(x,\lambda)=\int^{x}_{-\infty}e^{i\lambda(x-y)}q^{*}(y)\phi_{2}(y,\lambda)dy,
\tag{\theequation a}\\
&\phi_{2}(x,\lambda)=e^{-i\lambda x}
+\int^{x}_{-\infty}e^{-i\lambda(x-y)}q(y)\phi_{1}(y,\lambda)dy,\tag{\theequation b}\\
&\bar{\phi}_{1}(x,\lambda)=e^{i\lambda x}
+\int^{x}_{-\infty}e^{i\lambda(x-y)}q^{*}(y)\bar{\phi}_{2}(y,\lambda)dy,\tag{\theequation c}\\
&\bar{\phi}_{2}(x,\lambda)=\int^{x}_{-\infty}e^{-i\lambda(x-y)}q^{(y)}\bar{\phi}_{1}(y,\lambda)dy,
\tag{\theequation d}
\end{align}
and
\addtocounter {equation}{1}
\begin{align}
&\psi_{1}(x,\lambda)=e^{i\lambda x}
-\int^{\infty}_{x}e^{i\lambda(x-y)}q^{*}(y)\psi_{2}(y,\lambda)dy,\tag{\theequation a}\label{a}\\
&\psi_{2}(x,\lambda)=
-\int^{\infty}_{x}e^{-i\lambda(x-y)}q(y)\psi_{1}(y,\lambda)dy,\tag{\theequation b}\label{b}\\
&\bar{\psi}_{1}(x,\lambda)=
-\int^{\infty}_{x}e^{i\lambda(x-y)}q^{*}(y)\bar{\psi}_{2}(y,\lambda)dy,\tag{\theequation c}\\
&\bar{\psi}_{2}(x,\lambda)=e^{-i\lambda x}
-\int^{\infty}_{x}e^{-i\lambda(x-y)}q(y)\bar{\psi}_{1}(y,\lambda)dy.\tag{\theequation d}
\end{align}
\end{thm}

\begin{proof}
Here we only prove \eqref{a} and \eqref{b}, others can be similarly proved.

\textbf{First, we prove the existence of the Jost functions.}

To begin with, we make a change as follows
\begin{align}
f(x,\lambda)=e^{-i\lambda x}\psi(x,\lambda),
\end{align}
then \eqref{a} and \eqref{b} can be changed into
\begin{align}
&f_{1}(x,\lambda)=1-\int^{\infty}_{x}q^{*}(y)f_{2}(y,\lambda)dy,\\
&f_{2}(x,\lambda)=-\int^{\infty}_{x}e^{-2i\lambda (x-y)}q(y)f_{1}(y,\lambda)dy.\label{p2}
\end{align}
Eliminating the function $f_{2}$  from $f_{1}$, we can get
\begin{align}\label{p1}
f_{1}(x,\lambda)=1+\int^{\infty}_{x}q^{*}(y)
\left(\int^{\infty}_{y}e^{-2i\lambda (y-z)}q(z)f_{1}(z,\lambda)dz\right)dy.
\end{align}
Using successive approximation method yields
\begin{align}
&f^{(0)}_{1}(x,\lambda)=1,\\
&f^{(j)}_{1}(x,\lambda)=\int^{\infty}_{x}q^{*}(y)
\left(\int^{\infty}_{y}e^{-2i\lambda (y-z)}q(z)f^{(j-1)}_{1}(z,\lambda)dz\right)dy,\quad
j=1,2,\ldots.
\end{align}

Introducing
\begin{align*}
Q(x)=\int^{\infty}_{x}|q(y)|dy, \quad Q^{*}(x)=\int^{\infty}_{x}|q^{*}(y)|dy=Q(x),
\end{align*}
when $Im \lambda\geq0$, we have
\begin{align}
|e^{-2i\lambda (y-z)}|=|e^{2(y-z)Im\lambda-2i(y-z)Re\lambda}|=e^{2(y-z)Im\lambda}\leq1.
\end{align}
Based on the above conclusions, when $j=1$ we have
\begin{align*}
|f^{(1)}_{1}(x,\lambda)|
&=\left|\int^{\infty}_{x}q^{*}(y)\left(\int^{\infty}_{y}e^{-2i\lambda (y-z)}q(z)dz\right)dy\right|\\
&\leq\int^{\infty}_{x}\left|q^{*}(y)\right|\left(\int^{\infty}_{y}\left|e^{-2i\lambda (y-z)}\right|\left|q(z)\right|dz\right)dy\\
&\leq\int^{\infty}_{x}\left|q^{*}(y)\right|\left(\int^{\infty}_{y}\left|q(z)\right|dz\right)dy\\
&\leq\int^{\infty}_{x}\left|q^{*}(y)\right|dy\int^{\infty}_{x}\left|q(z)\right|dz\\
&=Q(x)Q^{*}(x).
\end{align*}
On the same way, we can see when $j=2$
\begin{align*}
|f^{(2)}_{1}(x,\lambda)|
&=\left|\int^{\infty}_{x}q^{*}(y)
\left(\int^{\infty}_{y}e^{-2i\lambda (y-z)}q(z)f^{(1)}_{1}(z,\lambda)dz\right)dy\right|\\
&\leq\int^{\infty}_{x}\left|q^{*}(y)\right|
\left(\int^{\infty}_{y}\left|q(z)\right|
\left(\int^{\infty}_{z}\left|q^{*}(s)\right|ds\int^{\infty}_{z}\left|q(s)\right|ds\right)
dz\right)dy\\
&\leq\int^{\infty}_{x}\left|q^{*}(y)\right|\int^{\infty}_{y}\left|q^{*}(s)\right|ds
\left(\int^{\infty}_{y}\left|q(z)\right|\left(\int^{\infty}_{z}\left|q(s)\right|ds\right)
dz\right)dy\\
&=\int^{\infty}_{x}\left|q^{*}(y)\right|\int^{\infty}_{y}\left|q^{*}(s)\right|ds
\left(\int^{\infty}_{y}-\frac{1}{2}\frac{d}{dz}
\left(\int^{\infty}_{z}\left|q(s)\right|ds\right)^{2}dz\right)dy\\
&=\frac{1}{2}\int^{\infty}_{x}\left|q^{*}(y)\right|\int^{\infty}_{y}\left|q^{*}(s)\right|ds
\left(\int^{\infty}_{y}\left|q(s)\right|ds\right)^{2}dy\\
&\leq\frac{1}{2}\left(\int^{\infty}_{x}\left|q(s)\right|ds\right)^{2}
\int^{\infty}_{x}\left|q^{*}(y)\right|\int^{\infty}_{y}\left|q^{*}(s)\right|dsdy\\
&=\frac{1}{2}\left(\int^{\infty}_{x}\left|q(s)\right|ds\right)^{2}
\frac{1}{2}\left(\int^{\infty}_{x}\left|q^{*}(s)\right|ds\right)^{2}\\
&=\frac{1}{2!^{2}}\left(Q(x)Q^{*}(x)\right)^{2}.
\end{align*}
Then we suppose the conclusion still holds when $j-1$, i.e.,
\begin{align*}
|f^{(j-1)}_{1}(x,\lambda)|=\frac{1}{(j-1)!^{2}}\left(Q(x)Q^{*}(x)\right)^{j-1},
\end{align*}
therefore
\begin{align*}
|f^{(j)}_{1}(x,\lambda)|
&=\left|\int^{\infty}_{x}q^{*}(y)
\left(\int^{\infty}_{y}e^{-2i\lambda (y-z)}q(z)f^{(j-1)}_{1}(z,\lambda)dz\right)dy\right|\\
&\leq\int^{\infty}_{x}\left|q^{*}(y)\right|\left(\int^{\infty}_{y}\left|q(z)\right|
\frac{1}{(j-1)!^{2}}\left(Q(z)Q^{*}(z)\right)^{j-1}dz\right)dy\\
&\leq\int^{\infty}_{x}\left|q^{*}(y)\right|\frac{Q^{*,j-1}(y)}{(j-1)!}
\left(\int^{\infty}_{y}\left|q(z)\right|\frac{Q^{j-1}(z)}{(j-1)!}dz\right)dy\\
&=\int^{\infty}_{x}\left|q^{*}(y)\right|\frac{Q^{*,j-1}(y)}{(j-1)!}
\left(\int^{\infty}_{y}\frac{d}{dz}\left(-\frac{1}{j!}Q^{j}(z)\right)dz\right)dy\\
&=\frac{1}{j!}\int^{\infty}_{x}\left|q^{*}(y)\right|\frac{Q^{*,j-1}(y)}{(j-1)!}Q^{j}(y)dy\\
&\leq\frac{1}{j!}Q^{j}(x)\int^{\infty}_{x}\left|q^{*}(y)\right|\frac{Q^{*,j-1}(y)}{(j-1)!}dy\\
&=\frac{1}{j!}Q^{j}(x)\frac{1}{j!}Q^{*,j}(x)\\
&=\frac{1}{j!}\left(Q(x)Q^{*}(x)\right)^{j}.
\end{align*}
So far, we can get by summarizing the above process
\begin{align}\label{p3}
|f^{(j)}_{1}(x,\lambda)|\leq\frac{1}{j!}\left(Q(x)Q^{*}(x)\right)^{j}
=\frac{1}{j!}\left(Q(x)\right)^{2j}.
\end{align}
Noticing
\begin{align*}
|f_{1}(x,\lambda)|
&=|f^{(0)}_{1}(x,\lambda)+f^{(1)}_{1}(x,\lambda)+\cdots+f^{(j)}_{1}(x,\lambda)+\cdots|\\
&\leq|f^{(0)}_{1}(x,\lambda)|+|f^{(1)}_{1}(x,\lambda)|+\cdots+|f^{(j)}_{1}(x,\lambda)|+\cdots\\
&=1+Q^{2}(x)+\frac{1}{2!^{2}}Q^{4}(x)+\cdots+\frac{1}{j!^{2}}Q^{2j}(x)+\cdots\\
&=\sum^{\infty}_{j=0}\frac{1}{j!^{2}}Q^{2j}(x)\\
&=\sum^{\infty}_{j=0}\frac{1}{j!\Gamma(j+0+1)}\left(\frac{2Q(x)}{2}\right)^{2j+0}\\
&=Bessel_{0}(2Q(x)).
\end{align*}
According to the properties of Bessel function, it is not difficult to verify that it is convergent, thus we can see $f_{1}(x,\lambda)$ absolutely and uniformly convergent.

On the other hand, we have
\begin{align*}
\sum^{\infty}_{j=0}f^{(j)}_{1}(x,\lambda)
=1+&\int^{\infty}_{x}q^{*}(y)\left(\int^{\infty}_{y}e^{-2i\lambda (y-z)}q(z)dz\right)dy\\
+&\int^{\infty}_{x}q^{*}(y)
\left(\int^{\infty}_{y}e^{-2i\lambda (y-z)}q(z)f^{(1)}_{1}(z,\lambda)dz\right)dy+\cdots\\
+&\int^{\infty}_{x}q^{*}(y)
\left(\int^{\infty}_{y}e^{-2i\lambda (y-z)}q(z)f^{(j-1)}_{1}(z,\lambda)dz\right)dy+\cdots\\
=1+&\int^{\infty}_{x}q^{*}(y)\left(\int^{\infty}_{y}e^{-2i\lambda (y-z)}q(z)
\left(1+f^{(1)}_{1}+\cdots+f^{(j-1)}_{1}+\cdots\right)dz\right)dy\\
=1+&\int^{\infty}_{x}q^{*}(y)\left(\int^{\infty}_{y}e^{-2i\lambda (y-z)}q(z)
\sum^{\infty}_{j=0}f^{(j)}_{1}(x,\lambda)dz\right)dy,
\end{align*}
that is to say, the function $f_{1}(x,\lambda)$ is a solution of \eqref{p1}. According to \eqref{p2}, we can calculate $f_{2}(x,\lambda)$, so the existence of Jost functions can be proved.

\textbf{Second, we prove the uniqueness of the Jost functions.}

Similarly to the above way, we only need to prove the uniqueness of $f_{1}(x,\lambda)$.

Assuming \eqref{p1} has another solution $\bar{f}_{1}(x,\lambda)$, we have
\begin{align*}
\left|f_{1}(x,\lambda)-\bar{f}_{1}(x,\lambda)\right|
&=\left|\int^{\infty}_{x}q^{*}(y)\left(\int^{\infty}_{y}e^{-2i\lambda (y-z)}q(z)\left(f_{1}(z,\lambda)-\bar{f}_{1}(x,\lambda)\right)dz\right)dy\right|\\
&\leq\int^{\infty}_{x}\left|q^{*}(y)\right|\left(\int^{\infty}_{y}\left|q(z)\right|
\left|\left(f_{1}(z,\lambda)-\bar{f}_{1}(x,\lambda)\right)\right|dz\right)dy\\
&\leq\int^{\infty}_{x}\left|q^{*}(y)\right|dy\int^{\infty}_{x}\left|q(z)\right|
\left|\left(f_{1}(z,\lambda)-\bar{f}_{1}(x,\lambda)\right)\right|dz\\
&\leq Q(x)\int^{\infty}_{x}\left|q(z)\right|
\left|\left(f_{1}(z,\lambda)-\bar{f}_{1}(x,\lambda)\right)\right|dz.
\end{align*}
Because of $Q(x)=\int^{\infty}_{x}\left|q^{*}(y)\right|dy\leq\int^{\infty}_{-\infty}\left|q^{*}(y)\right|dy$ and let $Q=\int^{\infty}_{-\infty}\left|q^{*}(y)\right|dy$,
then we introduce
\begin{align*}
&F(x,\lambda)=f_{1}(z,\lambda)-\bar{f}_{1}(x,\lambda),\\
&G(x,\lambda)=Q\int^{\infty}_{x}\left|q(z)\right|\left|F(z,\lambda)\right|dz,
\end{align*}
thus
\begin{align}\label{p2}
F(x,\lambda)\leq G(x,\lambda).
\end{align}
Let's take the derivative of two sides of $G(x,\lambda)$ about $x$, i.e.,
\begin{align*}
G_{x}(x,\lambda)=-Q\left|q(x)\right|\left|F(x,\lambda)\right|,
\end{align*}
compared with \eqref{p2},
\begin{align*}
G_{x}(x,\lambda)\geq-Q\left|q(x)\right|\left|G(x,\lambda)\right|,
\end{align*}
i.e.,
\begin{align*}
\left(G(x,\lambda)\cdot e^{-Q\int^{\infty}_{x}\left|q(y)\right|dy}\right)_{x}\geq0.
\end{align*}
From the above results, $G(x,\lambda)\cdot e^{-Q\int^{\infty}_{x}\left|q(y)\right|dy}$ is nonnegative and monotonically increasing. We can find $G(x,\lambda)=0$ as $x\rightarrow\infty$, which means
\begin{align*}
f_{1}(x,\lambda)=\bar{f}_{1}(x,\lambda).
\end{align*}
So far, the uniqueness of the Jost functions are proved.

\textbf{Last, we prove the differentiability of the Jost functions.}

According to iterative sequences of $f_{1}(x,\lambda)$, we can get
\begin{align*}
f^{(0)}_{1,x}(x,\lambda)=&0,\quad f^{(0)}_{1,\lambda}(x,\lambda)=0,\\
f^{(j)}_{1,x}(x,\lambda)=&-q^{*}(y)\left(\int^{\infty}_{y}e^{-2i\lambda (y-z)}q(z)f^{(j-1)}_{1}(z,\lambda)dz\right),\\
f^{(j)}_{1,\lambda}(x,\lambda)=&-2i\int^{\infty}_{x}q^{*}(y)
\left(\int^{\infty}_{y}e^{-2i\lambda (y-z)}q(z)f^{(j-1)}_{1}(z,\lambda)dz\right)dy\\
&+\int^{\infty}_{x}q^{*}(y)
\left(\int^{\infty}_{y}e^{-2i\lambda (y-z)}q(z)f^{(j-1)}_{1,\lambda}(z,\lambda)dz\right)dy.
\end{align*}
By using \eqref{p3}, it is easy to get
\begin{align}
&f^{(j)}_{1,x}(x,\lambda)\leq\frac{1}{(j-1)!^{2}}|q(x)|Q^{2j-1}(x),\\
&f^{(j)}_{1,\lambda}(x,\lambda)\leq\frac{8}{(j-1)!^{2}}Q^{2j-1}(x)
\int^{\infty}_{x}\left|yq(y)\right|dy.
\end{align}
We can see that $f_{1}(x,\lambda)$ is differentiable to $x$ and $\lambda$, as well as $\sum^{\infty}_{j=0}f^{(j)}_{1,x}$ and $\sum^{\infty}_{j=0}f^{(j)}_{1,\lambda}$ converges uniformly on $a\leq x\leq\infty$ when $Im\lambda\geq0$.
Obviously, $f_{1,x}(x,\lambda)\sim0$, $f_{1,\lambda}(x,\lambda)\sim0$ when $x\rightarrow\infty$.
Therefore, the proof of differentiability of function is finished.

In the same way, we can prove the other Jost functions.
\end{proof}

Moreover, it is easy to verify the following relationship
\begin{align}\label{q1}
\phi_{2}=\bar{\phi}^{*}_{1},\quad\quad\phi_{1}=\bar{\phi}^{*}_{2},\\
\psi_{2}=\bar{\psi}^{*}_{1},\quad\quad\psi_{1}=\bar{\psi}^{*}_{2}.
\end{align}
Furthermore, they satisfy the following asymptotic conditions
\addtocounter {equation}{1}
\begin{align}
&\phi(x,\lambda)\sim
\left(\begin{array}{c}
0\\
1\\
\end{array}
\right)e^{-i\lambda x},\quad
\bar{\phi}(x,\lambda)\sim
\left(\begin{array}{c}
1\\
0\\
\end{array}
\right)e^{i\lambda x},
\quad\quad
as\quad x\rightarrow -\infty,\tag{\theequation a}\label{Q5.1}\\
&\psi(x,\lambda)\sim
\left(\begin{array}{c}
1\\
0\\
\end{array}
\right)e^{i\lambda x},\quad
\bar{\psi}(x,\lambda)\sim
\left(\begin{array}{c}
0\\
1\\
\end{array}
\right)e^{-i\lambda x},
\quad\quad
as\quad  x\rightarrow +\infty.\tag{\theequation b}\label{Q5.2}
\end{align}
 The functions $\phi$ and $\psi$ allow analytic extensions to the upper half $\lambda$-plane $\mathbb{C}^{+}$, while $\bar{\phi}$ and $\bar{\psi}$ allow analytic extensions to the upper half $\lambda$-plane $\mathbb{C}^{-}$ except finite simple poles and continuous up to the real axis $Im\lambda=0$.

\subsection{The scatting data}

Defining Wronskian as $W[\phi(x,\lambda),\psi(x,\lambda)]=
\phi_{1}(x,\lambda)\psi_{2}(x,\lambda)-\phi_{2}(x,\lambda)\psi_{1}(x,\lambda)$, obviously,  Wronskian has the following relations
\begin{align}\label{Q6}
&W[\phi(x,\lambda),\psi(x,\lambda)]=-W[\psi(x,\lambda),\phi(x,\lambda)],\\
&W[c_{1}\phi(x,\lambda),c_{2}\psi(x,\lambda)]=c_{1}c_{2}W[\phi(x,\lambda),\psi(x,\lambda)].
\end{align}

According to boundary conditions, we can know
\begin{align}\label{Q7}
W[\phi(x,\lambda),\bar{\phi}(x,\lambda)]=-1=-W[\psi(x,\lambda),\bar{\psi}(x,\lambda)].
\end{align}

In fact, since Jost functions are solutions of spectral problems, they must be linearly depedent. There have functions $a(\lambda)$, $\bar{a}(\lambda)$, $b(\lambda)$ and $\bar{b}(\lambda)$ about $\lambda$ that makes
\begin{align}\label{Q8}
\left\{\begin{aligned}
\phi(x,\lambda)=a(\lambda)\bar{\psi}(x,\lambda)+b(\lambda)\psi(x,\lambda),\\
\bar{\phi}(x,\lambda)=\bar{a}(\lambda)\psi(x,\lambda)+\bar{b}(\lambda)\bar{\psi}(x,\lambda),
\end{aligned}\right.
\end{align}
and satisfy the following propositions.

\begin{prop} The functions
$a(\lambda)$, $\bar{a}(\lambda)$, $b(\lambda)$ and $\bar{b}(\lambda)$ about $\lambda$ satisfy the following relation on the real $\lambda$-axis
\begin{align}\label{Q9}
b(\lambda)\bar{b}(\lambda)-a(\lambda)\bar{a}(\lambda)=-1.
\end{align}
\end{prop}

\begin{proof}
Substituting \eqref{Q8} into Wronskian, one can see
\begin{align*}
W[\phi(x,\lambda),\bar{\phi}(x,\lambda)]
&=W[a(\lambda)\bar{\psi}(x,\lambda)+b(\lambda)\psi(x,\lambda),
\bar{a}(\lambda)\psi(x,\lambda)+\bar{b}(\lambda)\bar{\psi}(x,\lambda)]\\
&=\left|
   \begin{array}{cc}
     a(\lambda)\bar{\psi}_{1}(x,\lambda)+b(\lambda)\psi_{1}(x,\lambda) & \bar{a}(\lambda)\psi_{1}(x,\lambda)+\bar{b}(\lambda)\bar{\psi}_{1}(x,\lambda) \\
     a(\lambda)\bar{\psi}_{2}(x,\lambda)+b(\lambda)\psi_{2}(x,\lambda) & \bar{a}(\lambda)\psi_{2}(x,\lambda)+\bar{b}(\lambda)\bar{\psi}_{2}(x,\lambda) \\
   \end{array}
 \right|\\
&=\frac{b(\lambda)\bar{b}(\lambda)-a(\lambda)\bar{a}(\lambda)}{\bar{b}(\lambda)}
\left|
   \begin{array}{cc}
     \psi_{1}(x,\lambda) & \bar{a}(\lambda)\psi_{1}(x,\lambda)+\bar{b}(\lambda)\bar{\psi}_{1}(x,\lambda) \\
     \psi_{2}(x,\lambda) & \bar{a}(\lambda)\psi_{2}(x,\lambda)+\bar{b}(\lambda)\bar{\psi}_{2}(x,\lambda) \\
   \end{array}
 \right|\\
&=\left(b(\lambda)\bar{b}(\lambda)-a(\lambda)\bar{a}(\lambda)\right)
W[\psi(x,\lambda),\bar{\psi}(x,\lambda)],
\end{align*}
and because of \eqref{Q7}, we know that
\begin{align*}
b(\lambda)\bar{b}(\lambda)-a(\lambda)\bar{a}(\lambda)=-1.
\end{align*}
\end{proof}

\begin{prop}
According to \eqref{Q7} and \eqref{Q8}, the following results hold
\begin{align}
a(\lambda)=-W[\phi(x,\lambda),\psi(x,\lambda)],\quad\quad
\bar{a}(\lambda)=W[\bar{\phi}(x,\lambda),\bar{\psi}(x,\lambda)],\label{Q10.1}\\
b(\lambda)=W[\phi(x,\lambda),\bar{\psi}(x,\lambda)],\quad\quad
\bar{b}(\lambda)=-W[\bar{\phi}(x,\lambda),\psi(x,\lambda)].\label{Q10.2}
\end{align}
\end{prop}

\begin{proof}
According to the first expression of \eqref{Q8}, we have
\begin{align*}
\phi_{1}(x,\lambda)=a(\lambda)\bar{\psi}_{1}(x,\lambda)+b(\lambda)\psi_{1}(x,\lambda),\\
\phi_{2}(x,\lambda)=a(\lambda)\bar{\psi}_{2}(x,\lambda)+b(\lambda)\psi_{2}(x,\lambda).
\end{align*}
Then we use $\psi_{2}(x,\lambda)$ and  $\psi_{1}(x,\lambda)$ to multiply above two formulas respectively,
\begin{align*}
\phi_{1}(x,\lambda)\psi_{2}(x,\lambda)
=a(\lambda)\bar{\psi}_{1}(x,\lambda)\psi_{2}(x,\lambda)
+b(\lambda)\psi_{1}(x,\lambda)\psi_{2}(x,\lambda),\\
\phi_{2}(x,\lambda)\psi_{1}(x,\lambda)
=a(\lambda)\bar{\psi}_{2}(x,\lambda)\psi_{1}(x,\lambda)
+b(\lambda)\psi_{2}(x,\lambda)\psi_{1}(x,\lambda).
\end{align*}
Finally, the second expression is subtracted from the first one for obtained results. Thus we can see
\begin{align*}
a(\lambda)=-W[\phi(x,\lambda),\psi(x,\lambda)].
\end{align*}
Similarly, the other three expressions can be obtained.
\end{proof}

Combining with \eqref{q1}, we have the symmetric relations
\begin{align}\label{q2}
a(\lambda)=\bar{a}^{*}(\lambda),\quad b(\lambda)=\bar{b}^{*}(\lambda).
\end{align}

There are finite discrete spectrum points $\lambda_{j}$ and $\bar{\lambda}_{j}$ as the zeros of $a(\lambda)$ in the upper half $\lambda$-plane and lower half $\lambda$-plane respectively. This is because the result of Wronskian calculated by $\phi(x,\lambda_{j})$ and $\psi(x,\lambda_{j})$  is irreversible, which further explains the linear depedent between $\phi(x,\lambda_{j})$ and $\psi(x,\lambda_{j})$. Similarly, $\bar{\phi}(x,\bar{\lambda}_{j})$ and $\bar{\psi}(x,\bar{\lambda}_{j})$ are also linear depedent.
Therefore, it is easy to see that
\begin{align}\label{Q11}
\left\{\begin{aligned}
\phi(x,\lambda_{j})=b(\lambda_{j})\psi(x,\lambda_{j}),\\
\bar{\phi}(x,\bar{\lambda}_{j})=\bar{b}(\bar{\lambda}_{j})\bar{\psi}(x,\bar{\lambda}_{j}).
\end{aligned}\right.
\end{align}

For constructing the normalizing coefficients about discrete spectrum points $\lambda_{j}$, we derive the two sides of the spectral problem \eqref{Q3} about $\lambda$, then
\begin{align}\label{Q12}
(\partial-i\lambda)\psi_{1,\lambda}(x,\lambda)
=i\psi_{1}(x,\lambda)+q^{*}(x)\psi_{2,\lambda}(x,\lambda),\\
(\partial+i\lambda)\psi_{2,\lambda}(x,\lambda)
=-i\psi_{2}(x,\lambda)+q(x)\psi_{1,\lambda}(x,\lambda).
\end{align}
Through some simple calculation, we can get
\begin{align}\label{Q13.1}
\frac{d}{dx}W[\psi_{\lambda}(x,\lambda),\phi(x,\lambda)]
=i\left(\psi_{1}(x,\lambda)\phi_{2}(x,\lambda)+\psi_{2}(x,\lambda)\phi_{1}(x,\lambda)\right).
\end{align}
Similarly, we also can obtain
\begin{align}\label{Q13.2}
\frac{d}{dx}W[\psi(x,\lambda),\phi_{\lambda}(x,\lambda)]
=-i\left(\psi_{1}(x,\lambda)\phi_{2}(x,\lambda)+\psi_{2}(x,\lambda)\phi_{1}(x,\lambda)\right).
\end{align}
Then we do some simple integral operations on \eqref{Q13.1} and \eqref{Q13.2}, and make subtractions for the obtained results. Finally substituting them into the zeros of $a(\lambda_{j})$, and making $l\rightarrow\infty$, we have
\begin{align}\label{Q14}
a_{\lambda}(\lambda_{j})
=2ib(\lambda_{j})\int^{\infty}_{-\infty}\psi_{1}(x,\lambda_{j})\psi_{2}(x,\lambda_{j})dx,
\end{align}
which is equivalent to
\begin{align}
2\int^{\infty}_{-\infty}\psi_{1}(x,\lambda_{j})\psi_{2}(x,\lambda_{j})dx
=-\frac{ia_{\lambda}(\lambda_{j})}{b(\lambda_{j})}.
\end{align}
If $\lambda_{j}$ is simple zeros of $a(\lambda)$, there is a constant $c_{j}$ that makes
\begin{align}\label{Q15}
2\int^{\infty}_{-\infty}c^{2}_{j}\psi_{1}(x,\lambda_{j})\psi_{2}(x,\lambda_{j})dx=1,
\end{align}
then
\begin{align}
c^{2}_{j}=\frac{ib(\lambda_{j})}{a_{\lambda}(\lambda_{j})},
\end{align}
here we define $c_{j}$ as the normalizing coefficient.

In the same way, the normalizing coefficient about discrete spectrum points $\bar{\lambda}_{j}$ can be written as
\begin{align}
\bar{c}^{2}_{j}=\frac{i\bar{b}(\bar{\lambda}_{j})}{\bar{a}_{\lambda}(\bar{\lambda}_{j})}.
\end{align}

Furthermore, by using \eqref{q2}, the symmetry of discrete spectrum points satisfy
\begin{align}\label{q3}
\lambda^{*}_{j}=\bar{\lambda}_{j},
\end{align}
and the symmetry of the normalizing coefficients satisfy
\begin{align}\label{q4}
(c^{2}_{j})^{*}=\bar{c}^{2}_{j}.
\end{align}

Next, we introduce the transmission coefficients
\begin{align*}
T(\lambda)=\frac{1}{a(\lambda)},\bar{T}(\lambda)=\frac{1}{\bar{a}(\lambda)},
\end{align*}
 and the reflection coefficients
\begin{align*}
R(\lambda)=\frac{b(\lambda)}{a(\lambda)},\bar{R}(\lambda)
=\frac{\bar{b}(\lambda)}{\bar{a}(\lambda)}.
\end{align*}
Due to \eqref{Q8}, we have
\begin{align}
\left\{\begin{aligned}
T(\lambda)\phi(x,\lambda)=\bar{\psi}(x,\lambda)+R(\lambda)\psi(x,\lambda),\\
\bar{T}(\lambda)\bar{\phi}(x,\lambda)=\psi(x,\lambda)+\bar{R}(\lambda)\bar{\psi}(x,\lambda).
\end{aligned}\right.
\end{align}

Therefore, through the above calculation, the scattering data $S(\lambda)$ is as follows:
\begin{align}
&\left\{Im\lambda=0,R(\lambda),\lambda_{j},c^{2}_{j},j=1,2,\cdots,l\right\},\\
&\left\{Im\lambda=0,\bar{R}(\lambda),\bar{\lambda}_{j},\bar{c}^{2}_{j},j=1,2,\cdots,\bar{l}\right\},
\end{align}
where $\bar{l}=l$.

\section{Inverse scattering transform}

In this section, we will recover the potential $q(x)$ in the spectral problem by substituting the scattering data from the previous section into the GLM integral equations.

\subsection{Translation transformation}

Through translation transformation, the Jost functions in the spectrum problem can be expressed by some integral kernels and their corresponding integral expressions. Assuming that
\begin{align}
&\psi(x,\lambda)=\left(
                  \begin{array}{c}
                    1 \\
                    0 \\
                  \end{array}
                \right)
                e^{i\lambda x}+\int^{\infty}_{x}K(x,s)e^{i\lambda s}ds,\quad
                 K(x,s)=\left(
                  \begin{array}{c}
                    K_{1}(x,s) \\
                    K_{2}(x,s) \\
                  \end{array}
                \right),\label{Q16.1}\\
&\bar{\psi}(x,\lambda)=\left(
                  \begin{array}{c}
                    1 \\
                    0 \\
                  \end{array}
                \right)
                e^{-i\lambda x}+\int^{\infty}_{x}\bar{K}(x,s)e^{-i\lambda s}ds,\quad
                 K(x,s)=\left(
                  \begin{array}{c}
                    \bar{K}_{1}(x,s) \\
                    \bar{K}_{2}(x,s) \\
                  \end{array}
                \right).\label{Q16.2}
\end{align}
and according to \eqref{Q16.1}, we have
\begin{align*}
&\psi_{1,x}=i\lambda e^{i\lambda x}-K_{1}(x,x)e^{i\lambda x}+\int^{\infty}_{x}K_{1,x}(x,s)e^{i\lambda s}ds,\\
&\psi_{2,x}=-K_{2}(x,x)e^{i\lambda x}+\int^{\infty}_{x}K_{2,x}(x,s)e^{i\lambda s}ds.
\end{align*}
If $\psi(x,\lambda)$ and $\bar{\psi}(x,\lambda)$ are the solutions of the spectrum problem \eqref{Q3}, we can get
\begin{align*}
&i\lambda e^{i\lambda x}-K_{1}(x,x)e^{i\lambda x}+\int^{\infty}_{x}K_{1,x}(x,s)e^{i\lambda s}ds\\
&=i\lambda e^{i\lambda x}+i\lambda\int^{\infty}_{x}K_{1}(x,s)e^{i\lambda s}ds
+q^{*}(\lambda)\int^{\infty}_{x}K_{2}(x,s)e^{i\lambda s}ds,
\end{align*}
by substituting the above correlation expressions into \eqref{Q3},
then
\begin{align}\label{Q17}
\int^{\infty}_{x}e^{i\lambda s}\left[K_{1,x}(x,s)+K_{1,s}(x,s)-q^{*}(\lambda)K_{2}(x,s)\right]ds
-\lim_{s\rightarrow\infty}K_{2}(x,s)e^{i\lambda s}=0.
\end{align}
Similarly, we have
\begin{align}\label{Q18}
\begin{split}
\int^{\infty}_{x}e^{i\lambda s}&\left[K_{2,x}(x,s)-K_{2,s}(x,s)-q(\lambda)K_{1}(x,s)\right]ds\\
&-\left[q(x)+2K_{2}(x,x)\right]e^{i\lambda x}
+\lim_{s\rightarrow\infty}K_{1}(x,s)e^{i\lambda s}=0.
\end{split}
\end{align}
Therefore, according to \eqref{Q17} and \eqref{Q18} we can get
\begin{align}
&K_{1,x}(x,s)+K_{1,s}(x,s)-q^{*}(\lambda)K_{2}(x,s)=0,\\
&K_{2,x}(x,s)-K_{2,s}(x,s)-q(\lambda)K_{1}(x,s),\\
&q(x)+2K_{2}(x,x)=0,\label{Q19}\\
&\lim_{s\rightarrow\infty}K(x,s)=0.
\end{align}

Through the same method of calculation, one can obtain
\begin{align}
&\bar{K}_{1,x}(x,s)-\bar{K}_{1,s}(x,s)-q^{*}(\lambda)\bar{K}_{2}(x,s)=0,\\
&\bar{K}_{2,x}(x,s)-\bar{K}_{2,s}(x,s)-q(\lambda)\bar{K}_{1}(x,s),\\
&q^{*}(\lambda)=\bar{K}_{2}(x,s)=0,\\
&\lim_{s\rightarrow\infty}\bar{K}(x,s)=0.
\end{align}

On the other hand, $\phi(x,\lambda)$ and $\bar{\phi}(x,\lambda)$ can be assumed as
\begin{align}
&\phi(x,\lambda)=\left(
                  \begin{array}{c}
                    0 \\
                    1 \\
                  \end{array}
                \right)
                e^{-i\lambda x}+\int^{x}_{-\infty}H(x,s)e^{-i\lambda s}ds,\quad
                 H(x,s)=\left(
                  \begin{array}{c}
                    H_{1}(x,s) \\
                    H_{2}(x,s) \\
                  \end{array}
                \right),\\
&\bar{\phi}(x,\lambda)=\left(
                  \begin{array}{c}
                    1 \\
                    0 \\
                  \end{array}
                \right)
                e^{i\lambda x}+\int^{x}_{-\infty}\bar{H}(x,s)e^{i\lambda s}ds,\quad
                 \bar{H}(x,s)=\left(
                  \begin{array}{c}
                    \bar{H}_{1}(x,s) \\
                    \bar{H}_{2}(x,s) \\
                  \end{array}
                \right).
\end{align}

\subsection{The Gel'fand-Levitan-Marchenko(GLM) integral equations}

In fact, the GLM integral equations presented here are not only related to the translation transformation kernel mentioned earlier, but also to the scattering data obtained in the previous section. For this reason, the following theorems can be summarized.

\begin{thm}
According to the scattering data of the spectrum problem obtained above
\begin{align*}
&\left\{Im\lambda=0,R(\lambda),\lambda_{j},c^{2}_{j},j=1,2,\cdots,l\right\},\\
&\left\{Im\lambda=0,\bar{R}(\lambda),\bar{\lambda}_{j},\bar{c}^{2}_{j},j=1,2,\cdots,\bar{l}\right\},
\end{align*}
and assuming
\begin{align}
&F_{c}(x)=\frac{1}{2\pi}\int^{+\infty}_{-\infty}R(\lambda)e^{i\lambda x}d\lambda,\quad\quad
\bar{F}_{c}(x)=\frac{1}{2\pi}\int^{+\infty}_{-\infty}\bar{R}(\lambda)e^{-i\lambda x}d\lambda,\\
&F_{d}(x)=\sum^{N}_{j=1}c^{2}_{j}e^{i\lambda_{j}x},\quad\quad\quad\quad\quad\quad
\bar{F}_{d}(x)=\sum^{\bar{N}}_{j=1}\bar{c}^{2}_{j}e^{-i\bar{\lambda}_{j}x},\\
&F(x)=F_{c}(x)-F_{d}(x),\quad\quad\quad\quad\quad\bar{F}(x)=\bar{F}_{c}(x)-\bar{F}_{d}(x),
\end{align}
then GLM integral equations can be obtained
\begin{align}\label{Q21}
\bar{K}(x,y)+\left(\begin{array}{c}
                 1 \\
                 0 \\
               \end{array}\right)
             F(x+y)+\int^{\infty}_{x}K(x,s)F(s+y)ds=0,\\
K(x,y)+\left(
               \begin{array}{c}
                 0 \\
                 1 \\
               \end{array}
             \right)
             \bar{F}(x+y)+\int^{\infty}_{x}\bar{K}(x,s)\bar{F}(s+y)ds=0.
\end{align}
\end{thm}

\begin{proof}
To begin with, by doing some calculations for the first formula of \eqref{Q8} we can get
\begin{align*}
\left(\frac{1}{a(\lambda)}-1\right)\phi(x,\lambda)
=&R(\lambda)\left[\psi(x,\lambda)-\left(\begin{array}{c}
                                     1 \\
                                     0 \\
                                   \end{array}\right)
                                   e^{i\lambda x}
\right]
+R(\lambda)\left(\begin{array}{c}
                                     1 \\
                                     0 \\
                                   \end{array}\right)
                                   e^{i\lambda x}\\
&+\left[\psi(x,\lambda)-\left(\begin{array}{c}
                                     0 \\
                                     1 \\
                                   \end{array}\right)
                                   e^{-i\lambda x}
\right]
+\left[\left(\begin{array}{c}
 0 \\
 1 \\
\end{array}\right)
e^{-i\lambda x}-\phi(x,\lambda)
\right],
\end{align*}
then operating the inverse Fourier transform on both sides of the above expression and simplifying the obtained formula, we can obtain
\begin{align*}
&\frac{1}{2\pi}\int^{+\infty}_{-\infty}\left(\frac{1}{a(\lambda)}-1\right)\phi(x,\lambda)e^{i\lambda y}d\lambda\\
&=\frac{1}{2\pi}\int^{+\infty}_{-\infty}R(\lambda)\left[\int^{\infty}_{x}K(x,s)e^{i\lambda s}ds
\right]e^{i\lambda y}d\lambda
+\left(\begin{array}{c}
1 \\
0 \\
\end{array}\right)\frac{1}{2\pi}
\int^{+\infty}_{-\infty}R(\lambda)
e^{i\lambda (x+y)}d\lambda\\
&+\frac{1}{2\pi}\int^{+\infty}_{-\infty}\left[\int^{\infty}_{x}\bar{K}(x,s)e^{-i\lambda s}ds
\right]e^{i\lambda y}d\lambda
-\frac{1}{2\pi}\int^{+\infty}_{-\infty}\left[\int^{x}_{-\infty}H(x,s)e^{-i\lambda s}ds
\right]e^{i\lambda y}d\lambda,
\end{align*}
here the last term is equal to zero by Fourier transform. On the left side of the above equation, we use the residue theorem to integrate the large loop of the upper half plane
\begin{align*}
&\frac{1}{2\pi}\int^{+\infty}_{-\infty}\left(\frac{1}{a(\lambda)}-1\right)\phi(x,\lambda)e^{i\lambda y}d\lambda\\
&=iRes\left(\frac{1}{a(\lambda)}-1\right)\phi(x,\lambda)e^{i\lambda y}
=\sum^{N}_{j=1}\frac{ib(\lambda_{j})\psi(x,\lambda_{j})}{a_{\lambda}(\lambda_{j})}
e^{i\lambda_{j} y}\\
&=\sum^{N}_{j=1}c^{2}_{j}\left[\left(\begin{array}{c}
1 \\
0 \\
\end{array}\right)
e^{i\lambda_{j}(x+y)}
+\int^{\infty}_{x}K(x,s)e^{i\lambda_{j}(s+y)}ds
\right]\\
&=\left(\begin{array}{c}
1 \\
0 \\
\end{array}\right)\sum^{N}_{j=1}c^{2}_{j}e^{i\lambda_{j}(x+y)}
+\int^{\infty}_{x}K(x,s)\sum^{N}_{j=1}c^{2}_{j}e^{i\lambda_{j}(s+y)}ds.
\end{align*}
By summarizing the above calculation, the original formula can be changed into
\begin{align*}
&\left(\begin{array}{c}
1 \\
0 \\
\end{array}\right)\sum^{N}_{j=1}c^{2}_{j}e^{i\lambda_{j}(x+y)}
+\int^{\infty}_{x}K(x,s)\sum^{N}_{j=1}c^{2}_{j}e^{i\lambda_{j}(s+y)}ds\\
&=\frac{1}{2\pi}\int^{+\infty}_{-\infty}R(\lambda)\left[\int^{\infty}_{x}K(x,s)e^{i\lambda s}ds
\right]e^{i\lambda y}d\lambda\\
&+\left(\begin{array}{c}
1 \\
0 \\
\end{array}\right)\frac{1}{2\pi}
\int^{+\infty}_{-\infty}R(\lambda)
e^{i\lambda (x+y)}d\lambda
+\frac{1}{2\pi}\int^{+\infty}_{-\infty}\left[\int^{\infty}_{x}\bar{K}(x,s)e^{-i\lambda s}ds
\right]e^{i\lambda y}d\lambda.
\end{align*}
Then, assuming that
\begin{align*}
&F_{c}(x)=\frac{1}{2\pi}\int^{+\infty}_{-\infty}R(\lambda)e^{i\lambda x}d\lambda,\\
&F_{d}(x)=\sum^{N}_{j=1}c^{2}_{j}e^{i\lambda_{j}x},
\end{align*}
the above expression becomes
\begin{align*}
&\left(\begin{array}{c}
1 \\
0 \\
\end{array}\right)F_{d}(x+y)+\int^{\infty}_{x}K(x,s)F_{d}(s+y)ds\\
&=\int^{\infty}_{x}K(x,s)F_{c}(s+y)ds+\left(\begin{array}{c}
1 \\
0 \\
\end{array}\right)F_{c}(x+y)+\bar{K}(x,y).
\end{align*}
Through transposition of terms and assuming
\begin{align*}
F(x)=F_{c}(x)-F_{d}(x),
\end{align*}
one can obtain
\begin{align*}
\bar{K}(x,y)+\left(\begin{array}{c}
                 1 \\
                 0 \\
               \end{array}\right)
             F(x+y)+\int^{\infty}_{x}K(x,s)F(s+y)ds=0.
\end{align*}

On the same method, we also can get
\begin{align*}
K(x,y)+\left(
               \begin{array}{c}
                 0 \\
                 1 \\
               \end{array}
             \right)
             \bar{F}(x+y)+\int^{\infty}_{x}\bar{K}(x,s)\bar{F}(s+y)ds=0.
\end{align*}
\end{proof}

\subsection{The time evolution of scatting data}

When the initial condition $q(x,0)$ is given, its corresponding scattering data is $S(\lambda,0)$. However, the scattering data corresponding to $q(x,t)$ is $S(\lambda,t)$. In order to recover the potential, next we will find out the relationship between them.

To begin with, the part of lax pair about $t$ is added with $vI$, where $v$ is a parameter, i.e.,
\begin{align}\label{Q23}
\phi_{t}=\left(
           \begin{array}{cc}
             A & B^{*} \\
             B & -A \\
           \end{array}
         \right)
         \phi-vI\phi.
\end{align}
Owing to $q(x)$ and its derivatives with respect to $x$ approaches to zero rapidly as $|x|\rightarrow \infty$, then there are
\begin{align*}
\lim_{|x|\rightarrow \infty}\left(
           \begin{array}{cc}
             A & B^{*} \\
             B & -A \\
           \end{array}
         \right)
=\left(\begin{array}{cc}
\bar{A}(\lambda) & 0 \\
0 & -\bar{A}(\lambda) \\
\end{array}\right).
\end{align*}

When $x\rightarrow-\infty$, $\phi(x,\lambda)$ in the above expression is replaced by the asymptotic condition in \eqref{Q5.1}, hence $v=-\bar{A}(\lambda)$.

When $x\rightarrow\infty$, according to the same way and combining with $\phi=a\bar{\psi}+b\psi$, we can get
\begin{align}\label{Q25}
\left(\begin{array}{c}
a_{t}e^{i\lambda x} \\
b_{t}e^{-i\lambda x} \\
\end{array}\right)
=\left(\begin{array}{cc}
\bar{A}(\lambda)-v & 0 \\
0 & -\bar{A}(\lambda)-v \\
\end{array}\right)
\left[
a\left(\begin{array}{c}
1 \\
0 \\
\end{array}\right)e^{i\lambda x}
+b\left(\begin{array}{c}
0 \\
1 \\
\end{array}\right)e^{-i\lambda x}
\right].
\end{align}
In component form \eqref{Q25} yields
\begin{align*}
\left\{\begin{aligned}
&a_{t}e^{i\lambda x}=\left[\bar{A}(\lambda)-v\right]ae^{i\lambda x},\\
&b_{t}e^{-i\lambda x}=\left[-\bar{A}(\lambda)-v\right]be^{-i\lambda x},
\end{aligned}\right.
\end{align*}
thus
\begin{align*}
\left\{\begin{aligned}
&a_{t}=\left[\bar{A}(\lambda)-v\right]a=2\bar{A}(\lambda)a,\\
&b_{t}=\left[-\bar{A}(\lambda)-v\right]b=0.
\end{aligned}\right.
\end{align*}

Therefore, with the above result, the normalizing coefficient can be written in the form
\begin{align*}
(c_{j})^{2}_{t}&=\left(\frac{ib}{a_{\lambda}}\right)_{t}
=i\frac{b_{t}a_{\lambda}-b(a_{\lambda})_{t}}{(a_{\lambda})^{2}}
=-i\frac{b\cdot2\bar{A}(\lambda)a_{\lambda}}{(a_{\lambda})^{2}}\\
&=-2\bar{A}(\lambda)c^{2}_{j},
\end{align*}
thus
\begin{align*}
c^{2}_{j}(t)=e^{-2\bar{A}(\lambda)t}c^{2}_{j}(0).
\end{align*}
The reflection coefficient also can be obtained as
\begin{align*}
R_{t}&=\left(\frac{b}{a}\right)_{t}
=\frac{b_{t}a-ba_{t}}{a^{2}}
=-\frac{b\cdot2\bar{A}(\lambda)a}{a^{2}}\\
&=-2\bar{A}(\lambda)R,
\end{align*}
i.e.,
\begin{align*}
R(t)=e^{-2\bar{A}(\lambda)t}R(0).
\end{align*}
Like the above way, we can get
\begin{align*}
&\bar{c}^{2}_{j}(t)=e^{2\bar{A}(\lambda)t}\bar{c}^{2}_{j}(0),\\
&\bar{R}(t)=e^{2\bar{A}(\lambda)t}\bar{R}(0).
\end{align*}

\subsection{Multi-solition solutions}

In this section, we will take into account the inverse scattering problem of \eqref{Q1} with reflectionless potentials. Supposing $R(\lambda)=\bar{R}(\lambda)=0$, the GLM integral equations reduce to
\begin{align}\label{Q26}
\left\{\begin{aligned}
&\bar{K}(x,y)\left(\begin{array}{c}
                 1 \\
                 0 \\
               \end{array}\right)
             F_{d}(x+y)+\int^{\infty}_{x}K(x,s)F_{d}(s+y)ds=0,\\
&K(x,y)-\left(\begin{array}{c}
                 0 \\
                 1 \\
               \end{array}
             \right)
             \bar{F}_{d}(x+y)-\int^{\infty}_{x}\bar{K}(x,s)\bar{F}_{d}(s+y)ds=0.
\end{aligned}\right.
\end{align}

Removing $\bar{K}(x,y)$ from the degenerated GLM integral equations, one can obtain
\begin{align}
\begin{split}
K(x,y)&-\left(\begin{array}{c}
                 0 \\
                 1 \\
               \end{array}\right)
             \bar{F}_{d}(x+y)
-\left(\begin{array}{c}
                 1 \\
                 0 \\
               \end{array}\right)\int^{\infty}_{x}F_{d}(x+s)\bar{F}_{d}(s+y)ds\\
&-\int^{\infty}_{x}\left[\int^{\infty}_{x}K(x,z)F_{d}(z+s)dz\right]\bar{F}_{d}(s+y)ds=0.
\end{split}
\end{align}
In component form, it is equivalent to
\begin{align}\label{Q27}
\left\{\begin{aligned}
&K_{1}(x,y)
-\int^{\infty}_{x}F_{d}(x+s)\bar{F}_{d}(s+y)ds
-\int^{\infty}_{x}K_{1}(x,z)\int^{\infty}_{x}F_{d}(z+s)\bar{F}_{d}(s+y)dsdz=0,\\
&K_{2}(x,y)-\bar{F}_{d}(x+y)
-\int^{\infty}_{x}K_{2}(x,z)\int^{\infty}_{x}F_{d}(z+s)\bar{F}_{d}(s+y)dsdz=0,
\end{aligned}\right.
\end{align}
where
\begin{align*}
\int^{\infty}_{x}F_{d}(z+s)\bar{F}_{d}(s+y)ds
&=\int^{\infty}_{x}\sum^{N}_{j=1}c^{2}_{j}e^{i\lambda_{j}(z+s)}
\sum^{\bar{N}}_{p=1}\bar{c}^{2}_{p}e^{-i\bar{\lambda}_{p}(s+y)}ds\\
&=\sum^{N}_{j=1}\sum^{\bar{N}}_{p=1}\frac{ic^{2}_{j}\bar{c}^{2}_{p}}{\lambda_{j}-\bar{\lambda}_{p}}
e^{i\lambda_{j}(z+x)-i\bar{\lambda}_{p}(x+y)}.
\end{align*}

To construct the multi-solition solution of the equation \eqref{Q1}, we introduce $N$-dimensional unit matrix $I_{N}$ and $\bar{N}$-dimensional unit matrix $I_{\bar{N}}$, as well as take $N\times1$ column vectors
\begin{align*}
g(x,t)=\left(g_{1}(x,t),g_{2}(x,t),\cdots,g_{N}(x,t),\right)^{T},\\
h(x,t)=\left(h_{1}(x,t),h_{2}(x,t),\cdots,h_{N}(x,t),\right)^{T},
\end{align*}
 $\bar{N}\times1$ column vectors
\begin{align*}
\bar{g}(x,t)=\left(\bar{g}_{1}(x,t),\bar{g}_{2}(x,t),\cdots,\bar{g}_{\bar{N}}(x,t),\right)^{T},\\
\bar{h}(x,t)=\left(\bar{h}_{1}(x,t),\bar{h}_{2}(x,t),\cdots,\bar{h}_{\bar{N}}(x,t),\right)^{T},
\end{align*}
and matrix $E(x,t)=(e_{pj})_{\bar{N}\times N}$, where
\begin{align*}
h_{m}(x,t)=c_{m}(t)e^{i\lambda_{m}x},\quad\quad
\bar{h}_{n}(x,t)=\bar{c}_{n}(t)e^{i\bar{\lambda}_{n}x},\quad\quad
e_{nm}(x,t)=\frac{h_{m}(x,t)\bar{h}_{n}(x,t)}{\lambda_{m}-\bar{\lambda}_{n}}.
\end{align*}

Supposing
\begin{align}\label{Q28}
K_{2}(x,y)=\bar{h}(y,t)^{T}\bar{g}(x,t),
\end{align}
and substituting \eqref{Q28} into the second expression of \eqref{Q27}, we   directly have
\begin{align}\label{Q29}
\begin{split}
\sum^{\bar{N}}_{n=1}&\bar{c}_{n}(t)\bar{g}_{n}(x,t)e^{-i\bar{\lambda}_{n}y}
-\sum^{\bar{N}}_{n=1}\bar{c}^{2}_{n}(t)e^{-i\bar{\lambda}_{n}(x+y)}\\
&-\sum^{\bar{N}}_{p=1}\sum^{N}_{j=1}\sum^{\bar{N}}_{n=1}\int^{\infty}_{x}
\bar{c}_{p}(t)\bar{g}_{p}(x,t)e^{-i\bar{\lambda}_{p}z}
\frac{ic^{2}_{j}(t)\bar{c}^{2}_{n}(t)}{\lambda_{j}-\bar{\lambda}_{n}}
e^{i\lambda_{j}(x+z)-i\bar{\lambda}_{n}(x+y)}dz=0,
\end{split}
\end{align}
here
\begin{align*}
\int^{\infty}_{x}e^{-i\bar{\lambda}_{p}z+i\lambda_{j}z}dz
=\frac{e^{-i\bar{\lambda}_{p}x+i\lambda_{j}x}}{-i(\lambda_{j}-\bar{\lambda}_{p})}.
\end{align*}
Hence \eqref{Q28} can be simplified as follows
\begin{align}
\bar{g}_{n}(x,t)-\bar{c}_{n}(t)e^{-i\bar{\lambda}_{n}x}
+\sum^{\bar{N}}_{p=1}\sum^{N}_{j=1}\bar{g}_{p}(x,t)e_{pj}(x,t)e_{nj}(x,t)=0,\quad
n=1,2,\cdots,\bar{N}.
\end{align}
In a vector form, it is equivalent to
\begin{align}
\bar{g}(x,t)-\bar{h}(x,t)+E(x,t)E(x,t)^{T}\bar{g}(x,t)=0,
\end{align}
which can be written as follows
\begin{align}
\bar{g}(x,t)=\left(I_{\bar{N}}+E(x,t)E(x,t)^{T}\right)^{-1}\bar{h}(x,t).
\end{align}
As a result, we have
\begin{align*}
K_{2}(x,y)&=\bar{h}(y,t)^{T}\bar{g}(x,t)\\
&=\bar{h}(y,t)^{T}\left(I_{\bar{N}}+E(x,t)E(x,t)^{T}\right)^{-1}\bar{h}(x,t)\\
&=tr\left[\left(I_{\bar{N}}+E(x,t)E(x,t)^{T}\right)^{-1}\bar{h}(x,t)\bar{h}(y,t)^{T}\right].
\end{align*}
Resorting to \eqref{Q19}, the potential $q(x)$ is recovered as follows
\begin{align}
\begin{split}
q(x)&=-2K_{2}(x,x)\\
&=-2tr\left[\left(I_{\bar{N}}+E(x,t)E(x,t)^{T}\right)^{-1}\bar{h}(x,t)\bar{h}(x,t)^{T}\right].
\end{split}
\end{align}

\section{Some solutions and their analysis}

In this section, we will give some soliton solutions and breathe solutions, and analyze their figures.

\subsection{One-soliton solution}

Take $N=\bar{N}=1$, one can see that
\begin{align*}
E(x,t)=e_{11}(x,t)=\frac{h_{1}\bar{h}_{1}}{\lambda_{1}-\bar{\lambda}_{1}},
\end{align*}
then
\begin{align}\label{Q30}
\begin{split}
q_{1}(x,t)
&=-2tr\left[\left(I_{\bar{N}}+E(x,t)E(x,t)^{T}\right)^{-1}\bar{h}(x,t)\bar{h}(x,t)^{T}\right]\\
&=-2(1+e^{2}_{11})^{-1}\bar{h}^{2}_{1}\\
&=\frac{(\lambda_{1}-\bar{\lambda}_{1})^{2}\bar{h}^{2}_{1}}
{(\lambda_{1}-\bar{\lambda}_{1})^{2}+(h_{1}\bar{h}_{1})^{2}}.
\end{split}
\end{align}

Setting $(\lambda_{1}-\bar{\lambda}_{1})^{2}+[c_{1}(0)]^{4}=0$ and $(\lambda_{1}-\bar{\lambda}_{1})^{2}+[\bar{c}_{1}(0)]^{4}=0$ as well as
denoting $\lambda_{1}=\alpha+i\beta$, substituting the above constraint and relevant assumptions involved into \eqref{Q30}, we obtain
\begin{align}
q_{1}(x,t)=\frac{4\alpha e^{2(-i(\alpha)^{2}-32\delta i(\alpha)^{6})t+2\alpha x}}
{-e^{2(i(\alpha)^{2}+32\delta i(\alpha)^{6})t-2\alpha x}e^{2(-i(\alpha)^{2}-32\delta i(\alpha)^{6})t+2\alpha x}+1}.
\end{align}

In what follows, we firstly discuss the case when $N=\bar{N}=1$. Fig. 1 shows the one-soliton solution by choosing the appropriate parameters. It is easy to find that real part and imaginary part of $\alpha$ have different effects on the solution respectively, while $\delta$ has little obvious effect.

\noindent
{\rotatebox{0}{\includegraphics[width=3.5cm,height=3.5cm,angle=0]{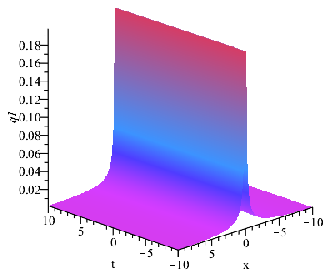}}}
~~~~
{\rotatebox{0}{\includegraphics[width=3.5cm,height=3.5cm,angle=0]{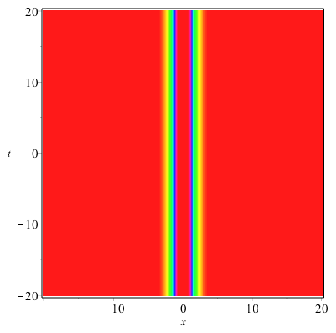}}}
\qquad\quad
{\rotatebox{0}{\includegraphics[width=3.5cm,height=3.5cm,angle=0]{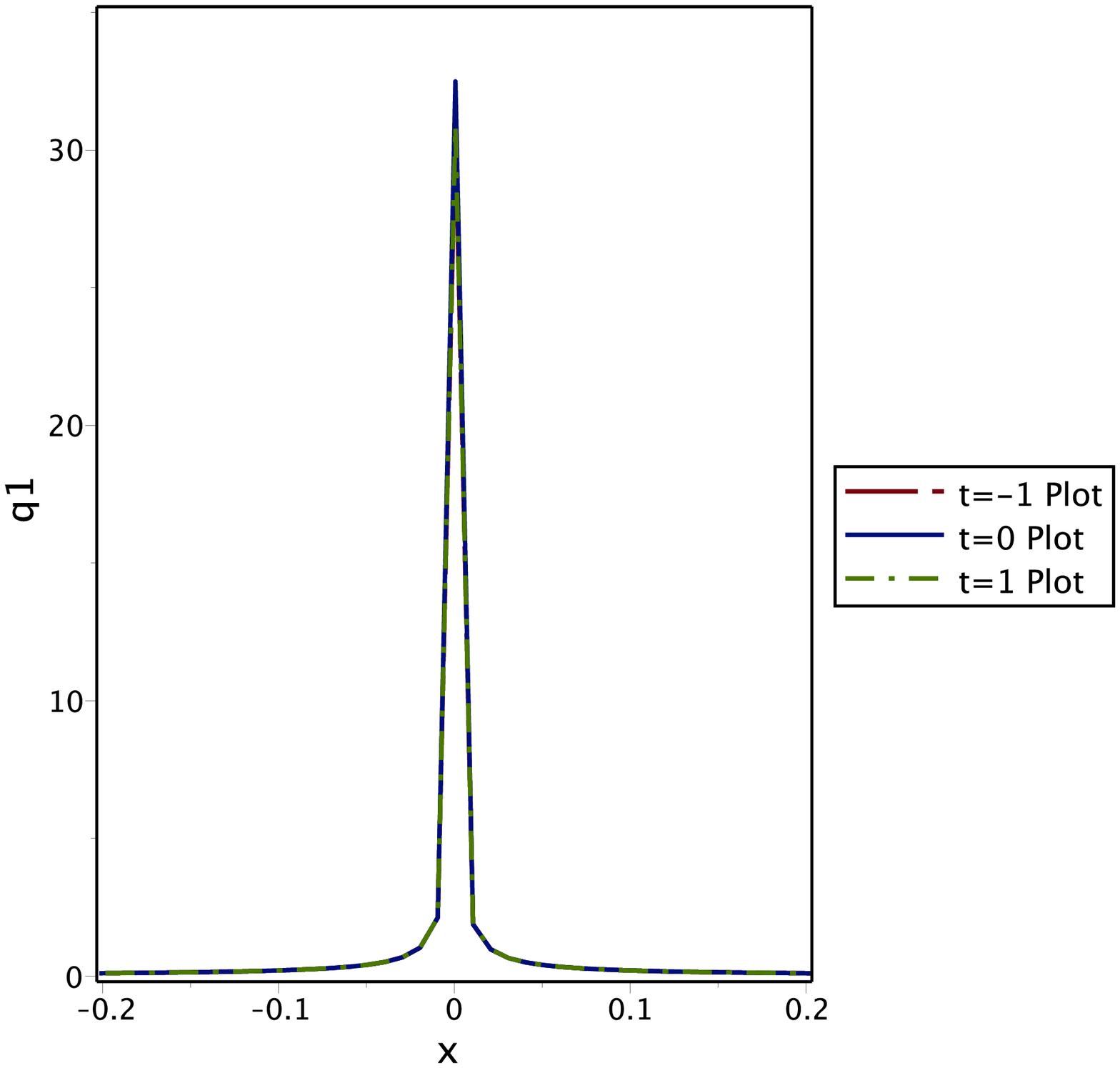}}}

\qquad\quad  $(a1)$
\qquad\qquad\qquad\qquad\qquad $(a2)$ \qquad\qquad\qquad\qquad\qquad$(a3)$\\
\noindent { \small \textbf{Figure 1.}
The single-soliton solutions for $|q_{1}|$ with the parameters selection
$\alpha=0.0001,\beta=-0.005,\delta=0.05.$
$\textbf{(a1)}$ three dimensional plot at time $t = 0$  in the $(x, t)$ plane,
$\textbf{(a2)}$ density plot,
$\textbf{(a3)}$ The wave propagation along the $x$-axis with $t = -1$, $t = 0$, $t = 1$.}

\subsection{Two-soliton solution}

When taking $N=2$, we denote
\begin{align*}
(\lambda_{1}-\bar{\lambda}_{1})^{2}+[c_{1}(0)]^{4}=0,\quad
(\lambda_{1}-\bar{\lambda}_{1})^{2}+[\bar{c}_{1}(0)]^{4}=0,\\
(\lambda_{2}-\bar{\lambda}_{2})^{2}+[c_{2}(0)]^{4}=0,\quad
(\lambda_{2}-\bar{\lambda}_{2})^{2}+[\bar{c}_{2}(0)]^{4}=0,
\end{align*}
and
\begin{align*}
\lambda_{1}=\alpha_{1}+i\beta_{1},\quad
\lambda_{2}=\alpha_{2}+i\beta_{2}.
\end{align*}

To begin with, we mark
\begin{align*}
W&=I_{2}+E(x,t)E(x,t)^{T}\\
&=\left(\begin{array}{cc}
      1 & 0 \\
      0 & 1 \\
    \end{array}\right)+
    \left(\begin{array}{cc}
      e_{11} & e_{12} \\
      e_{21} & e_{22} \\
    \end{array}\right)
    \left(\begin{array}{cc}
      e_{11} & e_{21} \\
      e_{12} & e_{22} \\
    \end{array}\right)\\
&=\left(\begin{array}{cc}
      1+e^{2}_{11}+e^{2}_{12} & e_{11}e_{21}+e_{12}e_{22} \\
      e_{11}e_{21}+e_{12}e_{22} & 1+e^{2}_{21}+e^{2}_{22} \\
    \end{array}\right).
\end{align*}
According to the expression of $W$, we have
\begin{align*}
W^{-1}=\frac{\left(\begin{array}{cc}
      1+e^{2}_{21}+e^{2}_{22} & -e_{11}e_{21}-e_{12}e_{22} \\
      -e_{11}e_{21}-e_{12}e_{22} & 1+e^{2}_{11}+e^{2}_{12} \\
    \end{array}\right)}
    {1+e^{2}_{11}+e^{2}_{12}+e^{2}_{21}+e^{2}_{22}+e^{2}_{11}e^{2}_{22}+e^{2}_{12}e^{2}_{21}
    -2e_{11}e_{12}e_{21}e_{22}}.
\end{align*}
Taking use of the above result, we can directly compute
\begin{align*}
&q_{2}(x,t)\\
&=-2tr\left[\left(I_{\bar{N}}+E(x,t)E(x,t)^{T}\right)^{-1}\bar{h}(x,t)\bar{h}(x,t)^{T}\right]\\
&=-2\frac{\bar{h}^{2}_{1}(1+e^{2}_{21}+e^{2}_{22})+\bar{h}^{2}_{2}(1+e^{2}_{11}+e^{2}_{12})
-2\bar{h}_{1}\bar{h}_{2}(e_{11}e_{21}+e_{12}e_{22})}
{1+e^{2}_{11}+e^{2}_{12}+e^{2}_{21}+e^{2}_{22}+e^{2}_{11}e^{2}_{22}+e^{2}_{12}e^{2}_{21}},
\end{align*}
where
\begin{align*}
&e_{11}=\frac{c_{1}(t)\bar{c}_{1}(t)e^{i(\lambda_{1}-\bar{\lambda}_{1})x}}{\lambda_{1}-\bar{\lambda}_{1}},
e_{12}=\frac{c_{2}(t)\bar{c}_{1}(t)e^{i(\lambda_{2}-\bar{\lambda}_{1})x}}{\lambda_{2}-\bar{\lambda}_{1}},\\
&e_{21}=\frac{c_{1}(t)\bar{c}_{2}(t)e^{i(\lambda_{1}-\bar{\lambda}_{2})x}}{\lambda_{1}-\bar{\lambda}_{2}},
e_{22}=\frac{c_{2}(t)\bar{c}_{2}(t)e^{i(\lambda_{2}-\bar{\lambda}_{2})x}}{\lambda_{2}-\bar{\lambda}_{2}},\\
&\bar{h}_{1}=\bar{c}_{1}(t)e^{-i\bar{\lambda}_{1}x},
\bar{h}_{2}=\bar{c}_{2}(t)e^{-i\bar{\lambda}_{2}x}.
\end{align*}

\noindent
{\rotatebox{0}{\includegraphics[width=3.5cm,height=3.5cm,angle=0]{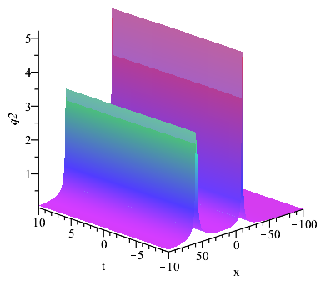}}}
~~~~
{\rotatebox{0}{\includegraphics[width=3.5cm,height=3.5cm,angle=0]{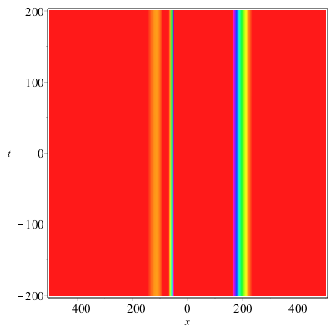}}}
\qquad\quad
{\rotatebox{0}{\includegraphics[width=3.5cm,height=3.5cm,angle=0]{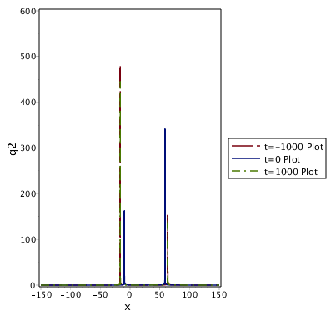}}}

\qquad\quad $(d1)$
\qquad\qquad\qquad\qquad\qquad $(d2)$ \qquad\qquad\qquad\qquad\qquad$(d3)$\\
\noindent { \small \textbf{Figure 2.}
The two-soliton solution for $|q_{2}|$ with the parameters selection
$\alpha_{1}=-0.01, \beta_{1}=-0.01,
\alpha_{2}=0.01, \beta_{2}=-0.01,
\delta=0.01.$
$\textbf{(d1)}$ three dimensional plot at time $t = 0$  in the $(x, t)$ plane,
$\textbf{(d2)}$ density plot,
$\textbf{(d3)}$ The wave propagation along the $x$-axis with $t = -1000$, $t = 0$, $t = 1000$.}\\

Next, we discuss the case for $N=\bar{N}=2$.
For the soliton solution  obtained in \cite{Sun-2017}, it is the soliton activity discussed under the condition of plane waves, and the two solitons collide. In Fig. 2, the two solitons moved side by side on the zero background, but the shape of the soliton is different at different times, so we guessed that the soliton influenced each other.

\section{Conclusions and discussions}

In this work, we have obtained multi-soliton solutions of the equations \eqref{Q1} according to the inverse scatting transform method. To start with, the Jost functions of spectrum problem \eqref{Q3} are derived and according to the result we have gotten the scatting data. Then integral kernels and their corresponding integral expressions have been given through translation transformation. It is easy to see that the relations of potential and integral kernels. Finally we have obtained multi-soliton solutions by GLM integral equations and the time evolution of scattering data. In addition, localized structures and dynamic behaviors of one-soliton and two-soliton solutions are illustrated vividly. It is hoped that our results can help enrich the nonlinear dynamics of the $N$-component nonlinear Schr\"{o}dinger type equations.

\section*{Acknowledgements}
This work was supported by  the Natural Science Foundation of Jiangsu Province under Grant No. BK20181351, the National Natural Science Foundation of China under Grant No. 11975306, the Six Talent Peaks Project in Jiangsu Province under Grant No. JY-059, the Qinglan Project of Jiangsu Province of China,  and the Fundamental Research Fund for the Central Universities under the Grant Nos. 2019ZDPY07 and 2019QNA35.


\begin{thebibliography}{3}
\bibitem{Zakharov-1972}
V.E. Zakharov, A.B. Shabat, Exact theory of two-dimensional self-focusing and one-dimensional self-modulation of waves in nonlinear media, Sov. Phys. JETP 34 (1972) 62-69.

\bibitem{Chowdury-2015}
A. Chowdury, D.J. Kedziora, A. Ankiewicz, and N. Akhmediev, Breather-to-soliton conversions described by the quintic equation of the nonlinear Schr\"{o}dinger hierarchy, Phys. Rev. E 91(3) (2015) 032928.
\bibitem{Chowdury-2015(2)}
A. Chowdury, A. Ankiewicz, and N. Akhmediev, Moving breathers and breather-to-soliton conversions for the Hirota equation, Proc. R. Soc. A 471(2180) (2015) 20150130.
\bibitem{Chen-2016}
S. Chen, F. Baronio, J.M. Soto-Crespo, Y. Liu, and P. Grelu, Chirped Peregrine solitons in a class of cubic-quintic nonlinear Schr\"{o}dinger equations, Phys. Rev. E 93(6) (2016) 062202.

\bibitem{Ankiewicz-2016}
A. Ankiewicz, D.J. Kedziora, A. Chowdury, U. Bandelow, and N. Akhmediev, Infinite hierarchy of nonlinear Schr\"{o}dinger equations and their solutions, Phys. Rev. E 93(1) (2016) 012206.

\bibitem{Sun-2017}
W.R. Sun, Breather-to-soliton transitions and nonlinear wave interactions for the nonlinear Schr\"{o}dinger equation with the sextic operators in optical fibers, Ann. Phys. 529 (2017) 1600227.

\bibitem{Gardner-1967}
C.S. Gardner, J.M. Greene, M.D. Kruskal, R.M. Miura, Method for the solving for the Korteweg-de
Veries equation, Phys. Rev. Lett. 19(19) (1967) 1095-1097.

\bibitem{Lax-1968}
P.D. Lax, Integrals of nonlinear equations of evolution and solitary waves, Comm. Pure. Appl. Math. 21 (1968) 467-490.

\bibitem{1}
M.J. Ablowitz, D.J. Kaup, A.C. Newell, H. Segur, Nonlinear-evolution equations of physical Significance, Phys. Rev. Lett. 31 (1973) 125-127.
\bibitem{2}
J.L. Ji, Z.N. Zhu, Soliton solutions of an integrable nonlocal modified Korteweg-de Vries equation through inverse scattering transform, J. Math. Anal. Appl. 453 (2017) 973-984.
\bibitem{3}
V. Caudrelier, Interplay between the Inverse Scattering Method and Fokas's unified transform with an application, Stud. Appl. Math. 140(1) (2018) 3-26.
\bibitem{RH-1}
S.F. Tian, Initial-boundary value problems for the general coupled nonlinear Schr\"{o}dinger
equation on the interval via the Fokas method, J. Differential Equations 262 (2017) 506-558.
\bibitem{Tian-jpa}
S.F. Tian, Initial-boundary value problems of the coupled modified Korteweg-de Vries equation on the half-line via the Fokas method, J. Phys. A: Math. Theor. 50(39) (2017) 395204.
\bibitem{NZBC-1}
B. Prinari, M.J. Ablowitz, and G. Biondini, Inverse scattering transform for the vector nonlinear Schr\"{o}dinger equation with nonvanishing boundary conditions, J. Math. Phys. 47 (2006) 063508.
\bibitem{NZBC-2}
M.J. Ablowitz, G. Biondini, and B. Prinari, Inverse scattering transform for the integrable discrete nonlinear Schr\"{o}dinger equation with nonvanishing boundary conditions, Inverse Prob. 23 (2007) 1711-1758.
\bibitem{NZBC-3}
B. Prinari, G. Biondini, and A.D. Trubatch, Inverse scattering transform for the multi-component nonlinear Schr\"{o}dinger equation with nonzero boundary conditions, Stud. Appl. Math. 126 (2011) 245-302.
\bibitem{NZBC-4}
F. Demontis, B. Prinari, C. van der Mee, and F. Vitale, The inverse scattering transform for the defocusing nonlinear Schr\"{o}dinger equations with nonzero boundary conditions, Stud. Appl. Math. 131 (2013) 1-40.
\bibitem{NZBC-5}
B. Prinari and F. Vitale, Inverse scattering transform for the focusing nonlinear Schr\"{o}dinger equation with one-sided nonzero boundary condition, Cont. Math. 651 (2015) 157-194.
\bibitem{Yan-2018}
G. Zhang,  Z. Yan, Inverse scattering transforms and $N$-double-pole solutions for the derivative NLS equation with zero/non-zero boundary conditions, arXiv:1812.02387.
\bibitem{Tian-NZyang}
W. Q. Peng, S. F. Tian, X. B. Wang, T. T. Zhang, Riemann-Hilbert method and
multi-soliton solutions for three-component coupled nonlinear Schr\"{o}dinger equations,
J. Geom. Phys. 146 (2019) 103508.
\bibitem{Biondini-2014}
G. Biondini,  G. Kova\u{c}i\u{c}, Inverse scattering transform for the focusing nonlinear Schr\"{o}dinger equation with nonzero boundary conditions, J. Math. Phys. 55(3) (2014) 031506.
\bibitem{15}
M.J. Ablowitz, D.B. Yaacov, A.S. Fokas, On the Inverse Scattering Transform for the Kadomtsev-Petviashvili Equation, Stud. Appl. Math. 69(2) (1983) 135-143.
\bibitem{17}
M. Wadati, The exact solution of the modified Korteweg-de Vries equation, J. Phys. Soc. Jpn. 32 (1972) 1681.
\bibitem{18}
J.J. Yang, S.F. Tian, W.Q. Peng and T.T. Zhang, The $N$-coupled higher-order nonlinear Schr\"{o}dinger equation: Riemann-Hilbert problem and multi-soliton solutions, Math. Meth. Appl. Sci.,  43(5) (2020) 2458-2472.
\bibitem{19}
S.V. Manakov, The inverse scattering transform for the time-dependent Schrodinger equation and
Kadomtsev-Petviashvili equation, Phys. D  3 (1981) 420-427.
\bibitem{Its-2011}
A. Its, E. Its, J. Kaplunov, Riemann-Hilbert Approach to the Elastodynamic Equation, Part I. Lett. Math. Phys. 96 (2011) 53-83.
\bibitem{Its-2018}
A. Its, E. Its, The Riemann-Hilbert approach to the Helmholtz equation in a quarter-plane: Neumann, Robin and Dirichlet boundary conditions, Lett. Math. Phys. 108 (2018) 1109-1135.
\bibitem{zhang-2004}
T.K. Ning, D.Y. Chen, D.J. Zhang, The exact solutions for the nonisospectral AKNS hierarchy
through the inverse scattering transform, Phys. A. 339 (2004) 248-266.
\bibitem{chen-2011}
X.M. Zhu, D.J. Zhang, D.Y. Chen, Lump Solutions of Kadomtsev-Petviashvili I Equation in Nonuniform Media, Commun. Theor. Phys. 55 (2011) 13-19.
\bibitem{Fokas-1997}
A.S. Fokas, A unified transform method for solving linear and certain nonlinear PDE¡¯s. Proc. R. Soc. Ser. A 453 (1997) 1411-1443.
\bibitem{Lenells-2016}
J. Lenells, The nonlinear steepest descent method: Asymptotics for initial-boundary value problems, SIAM J. Math. Anal. 48 (2016) 2076-2118.
\bibitem{Geng-2015}
X.G. Geng, H. Liu, J.Y. Zhu, Initial-boundary value problems for the coupled nonlinear Schr\"{o}dinger equation on the half-line, Stud. Appl. Math. (2015) 310-346.
\bibitem{Fan-2015}
J. Xu, E.G. Fan, Long-time asymptotics for the Fokas-Lenells equation with decaying initial value problem: Without solitons, J. Differential Equations 259 (2015) 1098-1148.
\bibitem{Fan-2015}
E.G. Fan, Y.C. Hon, Quasi-periodic waves and asymptotic behavior for the (2+1)-dimensional Bogoyavlenskii¡¯s breaking soliton equation, Phys. Rev. E. 78 (2008) 036607-036613.

\bibitem{xu}
T.Y. Xu, S.F. Tian, W.Q. Peng, Riemann-Hilbert approach for multisoliton solutions of generalized coupled fourth-order nonlinear Schr\"{o}dinger equations, Math. Meth. Appl. Sci. 43(2) (2020) 865-880.

\bibitem{20}
M.J. Ablowitz, H. Segur, Solitons and Inverse Scattering Transform, SIAM 4 (1981).







\end{thebibliography}
\end{document}